\documentclass[12pt]{article}

\usepackage[square,sort,comma]{natbib}
\usepackage{times} 
\usepackage{latexsym,epsfig,amssymb,amsmath,amsfonts,graphicx,amsthm,ifthen,pifont,comment,enumerate,setspace,multirow,color}
\usepackage{mhequ}
\usepackage{multicol,booktabs,colortbl,tabularx,graphicx}
\usepackage{caption}
\usepackage{subcaption}

\usepackage{JASA_manu}

\def\begmat{\left(\begin{array}}\def\endmat{\end{array}\right)}

\def\bi{\begin{itemize}\setlength{\itemsep}{0pt}} \def\ei{\end{itemize}}

\def\bl{\begin{list}{\labelitemi}{\leftmargin=1em}\setlength{\itemsep}{-2.5pt}}  \def\el{\end{list}}
\def\bn{\begin{enumerate}} \def\en{\end{enumerate}}
\def\bt{\begin{table}[h]} \def\et{\end{table}}
\def\bc{\begin{center}} \def\ec{\end{center}}

\newcommand{\abs}[1]{\left\vert#1\right\vert}

\newtheorem{theorem}{Theorem}[section]
\newtheorem{lemma}[theorem]{Lemma}

\theoremstyle{plain}

\theoremstyle{plain}

\theoremstyle{remark}

\theoremstyle{plain}

\newcommand \be{\begin{equs}}
\newcommand \ee{\end{equs}}

\begin{document}

\title{Variable Selection Using Shrinkage Priors}
\author{Hanning Li, \quad Debdeep Pati  \\
Department of Statistics  \\
Florida State University, Tallahassee, FL 32306 \\
email: \texttt{h.li@stat.fsu.edu, debdeep@stat.fsu.edu}}
\maketitle

\begin{center}
\textbf{Abstract}
\end{center}
Variable selection has received widespread attention over the last decade as we routinely encounter high-throughput datasets in complex biological and environment research.  Most Bayesian variable selection methods are restricted to mixture priors having separate components for characterizing the signal and the noise.
  However,  such priors encounter  computational issues in high dimensions.
  This has motivated continuous shrinkage priors, resembling the two-component priors facilitating computation and interpretability.  While such priors are widely used for estimating high-dimensional sparse vectors,  selecting a subset of variables remains a daunting task.  In this article, we propose a  general approach for variable selection with  shrinkage priors.  The presence of very few tuning parameters makes our method attractive in comparison to adhoc thresholding approaches.  The applicability of the approach is not limited to continuous shrinkage priors, but can be used along with any shrinkage prior.  Theoretical properties for near-collinear design matrices are investigated and the method is shown to have good performance in a wide range of synthetic data examples.

%

\vspace*{.3in}


\section{Introduction}
Variable selection in high-dimensional models has received substantial interest in recent years \citep{fan2010selective} and is a challenging problem for Bayesians.
With rapid advances in data acquisition and storage techniques, modern scientific investigations in epidemiology, genomics, imaging and networks are increasingly producing more variables compared to the sample size. One hope for meaningful inferences in such situations is to discover a subset of variables that explains the physical or biological process generating the data. Exploiting such underlying structures, commonly prevalent in the form of sparsity of model parameters, holds the key to meaningful inferences in high-dimensional settings.
This article revisits the problem of Bayesian variable selection in the context of Gaussian linear models \eqref{eq:model} using shrinkage priors:
\begin{eqnarray}\label{eq:model}
Y=X\beta+\epsilon, \quad \epsilon\sim N(0,\sigma^2 I_n),
\end{eqnarray}
where $Y$ is an $n$-dimensional response observed with respect to the $n\times p$ covariate matrix  $X$ and $\beta$ is the $p$-dimensional coefficient vector. Traditionally, to select the important variables out of $X_1, \ldots, X_p$,  a two component mixture prior (also referred to as a spike-and-slab prior)  \citep{mitchell1988bayesian,george1993variable,george1997approaches} is placed on $\beta$. These priors include a mass or a spike at zero characterizing the noise and a continuous density (usually centered at zero) representing the signal density.  
Although these priors are highly appealing in allowing separate control of the level of sparsity and the size of the signal coefficients, they lead to   computational hurdles in high-dimensions due to the need to explore a $2^p$ model space.  \cite{johnson2012nonlocal} recently showed a startling selection inconsistency phenomenon for several commonly used spike-and-slab priors based on intrinsic Bayes factors \citep{berger1996intrinsic}, fractional Bayes factors \citep{o1995fractional}, and g-priors \citep{liang2008mixtures} when the model size $p \gtrsim \sqrt{n}$.  This behavior was attributed to the common practice of centering the prior on the signal component at zero (local prior), which obliterates the demarcation between the signal and noise in high dimensions and leads to negligible posterior probability being assigned to any given model.   \cite{johnson2012nonlocal} advocated the use of non-local priors to obtain  selection consistency when $p=O(n)$, where the density for the signals decays to zero in a neighborhood of the origin.  When $p \gg n$,  it is not immediately clear whether non-local prior distributions can provide sufficient distinguishability between the signals and the noise coefficients.

Nevertheless, the practical problem of selecting variables has been a major bottleneck even with spike-and-slab priors. Although the highest posterior probability model (HPPM) is commonly perceived as the best model \citep{clyde1999bayesian,clyde1999empirical}, it is not optimal for prediction \citep{barbieri2004optimal} since HPPM is the Bayes estimate only under 0-1 loss function.  Moreover, finding HPPM  in high-dimensions  is  computationally demanding since the MCMC can only visit a minute fraction of the $2^p$ model space even for a relatively large number of Gibbs iterations.   To circumvent these issues, \cite{barbieri2004optimal} proposed the median probability model (MPM)  defined as the model consisting of those variables which have an overall posterior probability of inclusion greater than or equal to $1/2$. Although this is the optimal predictive model,  \cite{ghosh2014bayesian} found that summaries of the posterior distribution based on marginal and joint distributions may give conflicting results for assessing the importance of strongly correlated covariates.

Computational issues and considerations that many of the $\beta_j$s may be small but not
exactly zero have led to a rich variety of continuous shrinkage priors being proposed recently \citep{park2008bayesian,tipping2001sparse,griffin2010inference,carvalho2010horseshoe,carvalho2009handling,bhattacharya2014dirichlet}, which can be unified through a global-local (GL) scale mixture representation of \cite{polson2010shrink} below,
\begin{eqnarray}\label{eq:gl}
\beta_j \sim \mbox{N}(0, \lambda_j \tau), \quad \tau \sim f, \quad \lambda_j \sim g,
\end{eqnarray}
where $f$ and $g$ are densities on the positive real line.
In \eqref{eq:gl}, $\tau$ controls global shrinkage towards the origin while the local parameters $\lambda_j$s allow deviations in the degree of shrinkage. Special cases include Bayesian lasso \citep{park2008bayesian}, relevance vector machine \citep{tipping2001sparse}, normal-gamma mixtures \citep{griffin2010inference} and the horseshoe \citep{carvalho2010horseshoe,carvalho2009handling} among others.  GL priors potentially have substantial computational advantages over mixture priors, since the normal scale mixture representation allows for conjugate updating of $\beta$ and $\lambda$ in blocks. Moreover, a number of frequentist regularization procedures such as ridge, lasso, bridge and elastic net correspond to posterior modes under GL priors with appropriate choices of $f$ and $g$.

The literature on model selection with continuous shrinkage priors is even less-developed due to the unavailability of exact zeros in the posterior samples of $\beta$.  Heuristic methods based on thresholding the posterior mean/median of $\beta$ are often used in practice which lack theoretical justification, and inference is highly sensitive to the choice of the threshold.  There is a recent literature on decoupling shrinkage and selection \citep{hahn2015decoupling,bondell2012consistent,vehtari2002bayesian}, which poses the problem of selection as a loss function based decision rather than inducing sparsity through a prior distribution. 
Another naive way to select variables using a shrinkage prior is to see whether the posterior credible interval contains zero or not. Such a method usually has a poor performance because it is very difficult to estimate the uncertainty accurately in high dimensional problems.

In this article, we aim to address the problem of selecting variables through a novel method of post processing the posterior samples. The approach is based on first obtaining a posterior distribution of the number of signals by clustering the signal and the noise coefficients and then estimating the signals from the posterior median.  This simple approach requires very few tuning parameters and is shown to have excellent performance relative to both existing frequentist and Bayesian approaches.  Moreover, the method is not only applicable to continuous shrinkage priors, but also can be used along with any shrinkage prior for $\beta$ after a full MCMC run.   For the ease of exposition, we focus on the spike-and-slab prior and the horseshoe prior and compare the performances using HPPM, MPM and the credible set approach for variable selection.  Interestingly,  in the presence of high collinearity among the covariates, we demonstrated better performance  when the horseshoe prior is used in conjunction with our selection procedure.  


The organization of the present paper is as follows. Section \ref{sec:methods} describes the variable selection algorithm.  Theoretical properties for collinear design matrices are considered in Section \ref{sec:theory}.   Section \ref{sec:sims} contains detailed comparisons in synthetic data. A discussion is provided in Section \ref{sec:disc}.
\section{Methodology}\label{sec:methods}

Our objective is to develop an algorithm to select the important variables based on the posterior samples of $\beta$ obtained from the Markov Chain Monte Carlo (MCMC) samples  in \eqref{eq:model} when a  shrinkage prior is placed on $\beta$.  The algorithm is independent of the prior for $\beta$, but dependent on the linear model with additive error assumption in \eqref{eq:model}.   Unlike existing approaches, the method involves very few tuning parameters, hence readily suitable for future use of practitioners.  Our idea is based on finding the most probable set of variables in the posterior median of $\beta$.  Since the distribution of the number of important variables is more stable and is largely unaffected by the mixing of the MCMC,  we propose to first find the mode $H$ of the distribution of number of important variables and then select the $H$ largest coefficients from the posterior median of $\abs{\beta_j}$.

\subsection{2-means (2-M) variable selection}
We expect two clusters of $\abs{\beta_j}$, with one concentrated closely near zero corresponding to noise variables and the other one away from zero corresponding to the signals.  As an automated approach, we cluster $| \beta_j|$s at each MCMC iteration using {\em k-means} with $k=2$ clusters.  For the $i$th iteration, the number of non-zero signals $h_i$ is then estimated
by the smaller cluster size out of the two clusters. A final estimate ($H$) of the number of non-zero signals is obtained by taking the mode over all the MCMC iterations, i.e., $H=\mbox{mode}\{h_i\}$. The $H$ largest entries of the posterior median of $\abs{\beta}$ are identified as
 the non-zero signals.

When the true coefficient vector $\beta_T$ has signal coefficients varying in the signal strengths,  2-M variable selection approach described above may inappropriately cluster the smaller signals together with the noise variables. A possible solution is to use different values of the number of clusters $k$, but it is usually difficult and time-consuming to find the optimal value of $k$. To deal with this, we propose a simple modification called the {\em sequential 2-means} below.

\subsection{Sequential 2-means ($\mbox{S}_2\mbox{M}$) variable selection}

We start with a few notations.  Define $S_T$ to be the indices of non-zero signals in $\beta_T$ and $S_E$ to be the indices of the selected covariates.  To assess the efficacy of a variable selection procedure, we introduce two types of errors a) $\abs{S_T \cap S_E^c}$: masking error (also called `false negatives'), and b)  $\abs{S_E\cap S_T^c}$: swamping error (also called `false positives').

When $\beta_T$ has different levels of signal strengths, the 2-M variable selection approach will have a high chance of incurring masking error. In other words, it is highly likely that some true signals with low signal strengths will be clustered with the noise coefficients even when the corresponding $\beta_j$s are estimated well.
Our main motivation to propose the sequential 2-means ($\mbox{S}_2\mbox{M}$) variable selection approach is to reduce the probability of masking error.   Let $b > 0 $ be a tuning parameter, then $S_2M$ is defined as below. At the $i$th iteration of MCMC:
\begin{enumerate}[I]
\item perform a 2-means on $|\beta_j|, j=1,2,\ldots,p$. Denote the two cluster means by $m$ and $M$ ($m \leq M$). Initialize a set $A$ with an empty set. While the difference $M-m$ is greater than $b$:
\begin{enumerate}
\item update $A$ to be all the indices from the cluster with the lower mean $m$;
\item perform a 2-means on $\abs{\beta_j},j\in A$;
\item update $m$ and $M$ to be two cluster means ($m \leq M$) obtained from (b).
\end{enumerate}
\item The set $A$ is considered to contain coefficients of noise covariates. So the estimated number of signals $h_i$ is $p-\abs{A}$.
\end{enumerate}
The above algorithm is repeated for all MCMC samples of $\beta$ and the final estimates of the number of signals and the coefficients corresponding to signals and the noise variables are obtained in the same way as in the 2-M algorithm.

Using an appropriate tuning parameter $b$, $\mbox{S}_2\mbox{M}$ is capable of reducing the chance of masking error. However, note that this occurs at the cost of an increased probability of swamping error.  
A larger value of $b$ tends to increase  the masking error, while a smaller value of $b$ increases swamping error.
Hence one should choose $b$ so that the sum of the two errors is minimized.  In order to assess the important factors influencing the choice of $b$, let us first consider a noise-free version of the model in \eqref{eq:model} and hypothesise the ideal situation when the Bayesian procedure provides exactly accurate estimation of $\beta_T$ at every iteration. Then it is not hard to conclude any value of $b$ between 0 and the lowest absolute signal strength leads to correct variable selection. However, in presence of noise,  it is impossible for any method to produce an exact estimate of $\beta_T$. In addition, using continuous shrinkage priors, the estimated coefficients for the noise variables will never be exactly zero. As the noise level $\sigma$ increases,  the estimates for the coefficients become more variable.  This makes it more likely for the non-zero coefficients of lower signal strength to be clustered with the noise coefficients leading to an increase in the masking error.  Hence we  believe a proper value for $b$ should be  based on the posterior estimate of $\sigma$.  However, the estimate of $\sigma$ in \eqref{eq:model}  is affected by i) the true noise level ii) collinearity in the design matrix iii) as well as the {\em ill-posed}-ness  of the high-dimensional regression problem i.e., how large $p$ is compared to $n$.  These key factors contribute to the accuracy of the selection procedure.


The tuning parameter $b$ should be chosen to be an increasing function of the estimated $\sigma^2$ to take into account the increased variability in the estimates of the noise coefficients.
 The variable selection results using $\mbox{S}_2\mbox{M}$ approach with a  large threshold will surely be no worse than that the previously stated 2-M approach.
 Through various simulation with different settings, we have observed using $2$ times posterior median of $\sigma^2$ ($b=2\hat{\sigma}^2$) results in accurate estimation of the number of signals $H$.  Using $3\hat{\sigma}^2$ very often results in masking error,  while using $ \hat{\sigma}^2$ always selects many noise coefficients to be signals. In practice,  we suggest to use $b =  2\hat{\sigma}^2$ to reduce both masking and swamping errors.
 A higher value for $b$ might be necessary if the number of selected active covariates obviously exceeds the expected number of signals.  On the other hand,  a smaller value of $b$ would be desirable if  more covariates are expected to be active.


%

\section{Dealing with correlated predictors using continuous shrinkage priors}\label{sec:theory}
In presence of confounders which are highly correlated with an important predictor,  it is crucial that a variable selection method
 can identify the true predictor.   \cite{bhattacharya2014dirichlet,bhattacharya2012bayesian} recently showed that  a  global-local  shrinkage prior \eqref{eq:gl} achieves better concentration around sparse vectors in comparison with shrinkage prior based on only global scale i.e., setting $\psi_j \equiv 1$ in \eqref{eq:gl}.   In this section, we show that  such observations extend to the case of variable selection.  While a ``global-only'' shrinkage prior fails to select the true variables under moderate correlation, an appropriately constructed  global-local  shrinkage prior can achieve desirable variable selection even under high correlation.  

 For the ease of understanding the behavior of continuous shrinkage priors under correlation, consider only two covariates where variable 1 is the important predictor and variable 2 is the confounder.
 Let  $X'X = [1 \, \rho; \,\rho \, 1]$ with $\rho$ characterizing the correlation between the two predictors.  Assume $\hat{\beta}_{MLE, j}$ and $\hat{\beta}_{S, j}$ are the maximum likelihood estimate and  posterior mean of $\beta_j,  j =1,2$ respectively.   We empirically observe that MLE underestimates the signal coefficient $\beta_1$ and over-estimates the confounder $\beta_2$ under high correlation.  Ideally, a shrinkage prior should counter-balance this effect allowing the corresponding posterior estimates to be well-separated, thus facilitating variable selection. We define a terminology called {\em reverse-shrinkage}  to describe this.   A prior is said to satisfy reverse-shrinkage if  $|\hat{\beta}_{MLE,1}|\geq|\hat{\beta}_{MLE,2}|$ implies $|\hat{\beta}_{S,1} / \hat{\beta}_{S,2}| \geq |\hat{\beta}_{MLE,1}/ \hat{\beta}_{MLE,2}|$. Suppose we can write a Bayes estimator under a shrinkage prior  as a function of the MLE: $\hat\beta_{S,j}=(1-S_1(\hat\beta_{MLE}))\hat\beta_{MLE,j}$, $j=1,2,...,p$, where the shrinkage $S_j$ is less than 1. Then reverse-shrinkage means a larger magnitude of MLE results in a smaller shrinkage.  Clearly, this is a desirable phenomenon for any variable selection approach.  

\begin{theorem}\label{Theorem1}
Suppose $\beta_i\sim N(0,\sigma^2\tau^2), i=1,2$ in \eqref{eq:model}, and the $n \times 2$ covariate $X$ satisfies $X'X =  [1 \, \rho; \,\rho \, 1]$
where $\rho\in (0,1)$. Then if $|\beta_{MLE,1}|>|\beta_{MLE,2}|$,
\begin{eqnarray*}\label{theorem1res}
\abs{\frac{\hat\beta_{N,1}}{\hat\beta_{N,2}}}  <
\abs{\frac{\hat\beta_{MLE,1}}{\hat\beta_{MLE,2}}},
\end{eqnarray*}
for any $\tau$ and $n>2$, where $\hat\beta_N$ denotes the posterior mean.
\end{theorem}

Hence, when for correlated predictors, shrinkage priors with only global shrinkage parameters are no better than using MLEs. Next, we turn our attention to global-local shrinkage priors.  We focus on the horseshoe prior for a fixed value of the global shrinkage parameter $\tau$
\begin{eqnarray}\label{eq:HS}
\beta_j \mid \lambda_j,  \tau \sim N(0, \lambda_j \tau), \quad \lambda_j \sim \sqrt{\mbox{Ca}^+(0, 1)}
\end{eqnarray}
where $\mbox{Ca}^+(0, 1)$ denotes the standard half-Cauchy distribution with pdf  $2/ \{ \pi(1+ x^2)\}$ for $x > 0$.

With $X'X$ defined in Theorem \ref{Theorem1}, we write the HS estimators as functions of MLEs in Lemma \ref{hforhs} (see Appendix). More precisely, the HS estimators are functions of $\rho$, $\tau$, $\hat\beta_{MLE,2}$, and $A = \abs{\hat{\beta}_{MLE,1}/ \hat{\beta}_{MLE,2}}$. Given values for these parameters, we calculate the approximate values of $\abs{\hat\beta_{HS,1}/\hat\beta_{HS,2}}$ and $\abs{\hat\beta_{MLE,1}/\hat\beta_{MLE,2}}$ in Matlab to see whether reverse-shrinkage occurs.

Through the following figure, we will show the horseshoe prior can be made to satisfy  the reverse-shrinkage property by suitably choosing $\tau$.
 Figure \ref{fig:test} provides numerical analysis with $\hat\beta_{MLE,2} = 1$ and $1.5$, for different values of $\rho\in[0.94,0.99]$, $\tau\in(0,1)$, and $A>1$. Blue / red  dots indicate reverse-shrinkage / lack of it. Figure \ref{fig:test} shows reverse-shrinkage is more likely to occur when there are a smaller value of $\rho$, a greater value of $A$ (these two observations are expected) and a smaller value of $\tau$ present. The $\rho$ and $A$ are directly dependent on data, however, by choosing $\tau$ carefully chosen, it is possible to increase the possibility of achieving reverse-shrinkage.   Clearly, for values of $\rho$ close to 1, the horseshoe prior with large values of $\tau$ is less prone to achieve the reverse-shrinkage compared to smaller values of $\tau$.
 In practice, we suggest to have an upper bound for the global hyperparameter $\tau$ when updating it in a sampler.

\begin{figure}[htbp!]
\centering
\begin{subfigure}{.5\textwidth}
  \centering
 \includegraphics[width=9cm,height=6cm]{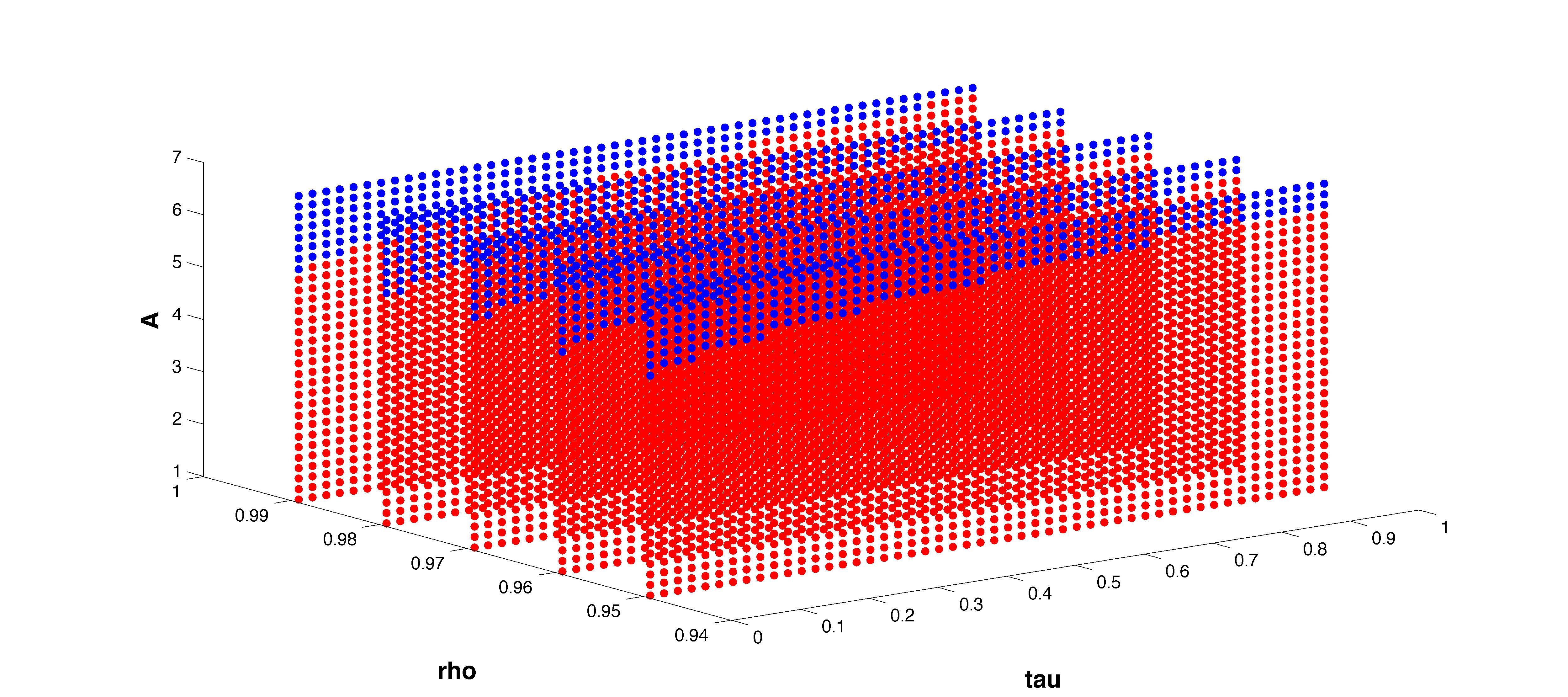}
  \caption{$\hat{\beta}_{MLE,2} = 1$}
  \label{fig:sub1}
\end{subfigure}%
\begin{subfigure}{.5\textwidth}
  \centering
 \includegraphics[width=9cm,height=6cm]{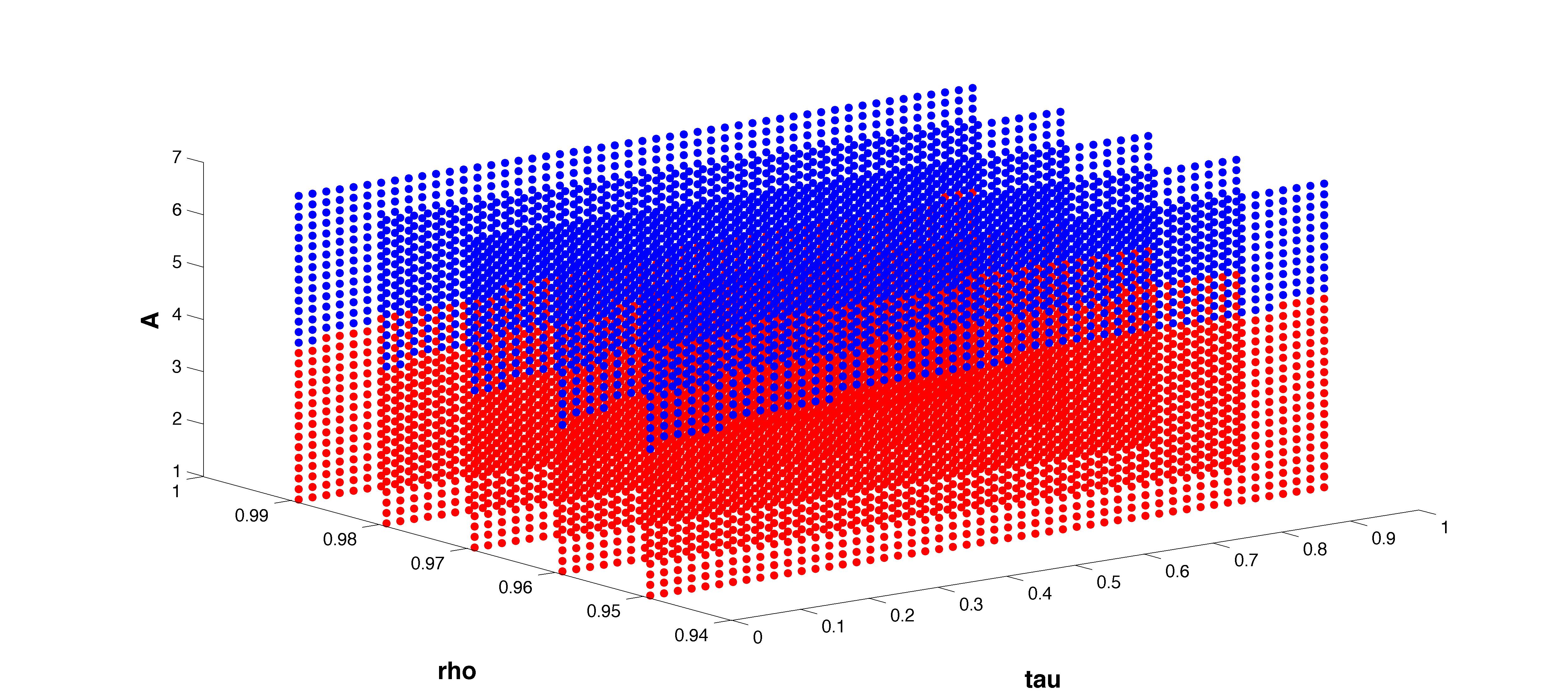}
  \caption{$\hat{\beta}_{MLE,2} = 1.5$ }
  \label{fig:sub2}
\end{subfigure}
\caption{+ve reverse-shrinkage (blue), -ve reverse-shrinkage (red) for various $A, \rho$ and $\tau$ }
\label{fig:test}
\end{figure}

\section{Simulation Study}\label{sec:sims}
 Our principle goal in this section is to compare the performance of the methods we proposed, i.e. $\mbox{S}_2\mbox{M}$ and 2-M, with other competing methods in terms of variable selection, especially when there is high collinearity present among the covariates. We consider the horseshoe (HS),  the spike-and-slab (SS) and the adaptive Lasso (AL) \citep{zou2006adaptive}.
 For SS, we used
 \begin{eqnarray*}
 \beta_j|\pi,\sigma_j^2,\sigma \sim \pi \delta_0  + (1- \pi) N(0, \sigma^2 \sigma_j^2), \quad \sigma_j^2 \sim \mbox{IG}(3/2, 3/2), \,1-\pi \sim \mbox{Beta}(1, 15), \,  \sigma^2 \sim \mbox{IG}(3/2, 3/2).
  \end{eqnarray*}
 The HS is used as in \cite{carvalho2010horseshoe}
 \begin{eqnarray}\label{eq:HS}
 \beta_j \mid \lambda_j, \tau, \sigma^2  \sim  \mbox{N}(0, \lambda_j \tau \sigma^2),  \quad \lambda_j \sim \sqrt{\mbox{Ca}^+(0, 1)},\,  \tau \sim \sqrt{\mbox{Ca}^+(0, 1)}, \, \sigma^2 \sim \mbox{IG}(3/2, 3/2).
 \end{eqnarray}
 \cite{carvalho2010horseshoe,carvalho2009handling}  suggested thresholding  posterior estimates of $\kappa_j = 1/ (1+ \lambda_j)$ at $1/2$ to perform  variable selection in HS. We refer to this as the Hypothesis Testing approach, abbreviated as HT. Although a suitable $\kappa_j$ can be defined in  the GL family \eqref{eq:gl},  we found HT to be most effective with HS.  It is important to reiterate here that $\mbox{S}_2\mbox{M}$ is independent of the prior for $\beta$.  We also used the credible set approach, abbreviated as CS. To implement AL,  we used the {\tt parcor} package in {\tt R}. For HS, the variable selection approaches tried are: $\mbox{S}_2\mbox{M}$, 2-M, CS and HT, while for SS we tried $\mbox{S}_2\mbox{M}$, 2-M, HPPM and MPM. In the first two examples, results are provided by averaging over 25 replicates of response obtained using the same covariate matrix.  For the Bayesian procedures, the MCMC was run for 5,000 iterations discarding a burn-in of 2,000. In all the examples, the tuning parameter $b$ is set to be $2\hat\sigma^2$. Convergence was monitored using standard tests and diagnostic trace plots.  In the following tables, the first number in each parenthesis is the masking error and the second number is the swamping error.

\subsection{Simulation example 1}
Four settings for sample size ($n$), the number of covariates ($p$), the number of signals ($r$) and signal strength ($B$) are (1) $n=50, p=300, r=10, B=4;$ (2)  $n=50, p=300, r=10, B=6$; (3) $n=100, p=800, r=20, B=4$ and (4) $n=100, p=800, r=20, B=6$. Under each setting, we considered an uncorrelated covariates setting ($uncor$) and a correlated covariates setting ($cor$).
Observations of the covariate are generated from standard normal distributions independently and an intercept is included.
The covariate matrix corresponding to $cor$ contains two pairs of correlated covariates, and in each pair, one is a signal while the other is a noise predictor. Both correlations are above 0.99.

In general,  HS out-performs  SS and AL significantly. We found the performance of our proposed method, $\mbox{S}_2\mbox{M}$ (or 2-M), is consistently better when used in conduction with HS, even when high collinearity is present between the covariates. CS often selects the wrong one between a highly correlated pair of covariates. When the difficulty of task is high, (see the 2nd and the 6th columns in Table.1), CS masked a larger fraction of the true signals. The performance of HS+HT is excellent as well. However, we have to note here HT is a variable selection method for priors from GL family only, while our proposed methods can be broadly applied with various priors.

\begin{table}[htb]
    \scriptsize
    \caption{Masking and Swamping errors}
    \centering

    \begin{tabular}{ccccccccc}
     \toprule
      $n,p,r$ & \multicolumn{4}{c}{$n=50, p=300,      r=10$}&
      \multicolumn{4}{c}{$n=100, p=800, r=20$}\\
      $B$& 4 & 4 & 6 & 6 & 4 & 4 & 6 & 6 \\
      & uncor & cor & uncor & cor & uncor & cor & uncor & cor \\
     \midrule
      HS+$\mbox{S}_2\mbox{M}$ &
      (0, 0)&(0.36, 0.36)&(0, 0)&(0.28, 0.28)&(0, 0)&(0.16, 0.16)&(0, 0)&(0.04, 0.04)\\
      HS+2M &(0, 0)&(0.36, 0.36)&(0, 0)&(0.28, 0.28)&(0, 0)&(0.16, 0.16)&(0, 0)&(0.04, 0.04)\\
      HS+CS &(0, 0)&(1.8, 0)&(0, 0)&(1.42, 0)&(0, 0)&(2.52, 0.02)&(0, 0)&(0.2, 0)\\
      HS+HT &(0, 0)&(0.4, 0.24)&(0, 0)&(0.24, 0.28)&(0, 0)&(0.2, 0.12)&(0, 0)&(0.04, 0.04)\\
      \midrule
      SS+$\mbox{S}_2\mbox{M}$
      &(0, 0)&(0.4, 0.4)&(0.6, 0.76)&(0.8, 1.08)&(0.8, 6.32)&(2.08, 5.56)&(4.84, 13.56)&(5.52, 16.32)\\
      SS+KM &(0, 0)&(0.4, 0.4)&(0.64, 0.48)&(0.8, 1.04)&(0.8, 6.32)&(1.96, 5.56)&(4.96, 13.44)&(5.48, 16.12)\\
      SS+HPPM &(0, 0.64)&(0.4, 1.32)&(0.48, 10.44)&(0.6, 4.2)&(4.96, 62.8)&(4.64, 57.36)&(7.96, 181.28)&(7.6, 179.6)\\
      SS+MPM &(0.04, 0.44)&(0.4, 0.8)&(0.52, 1.4)&(0.64, 1.72)&(4.16, 4.84)&(3.92, 4.96)&(7.2, 15.12)&(8, 15.8)\\
      \midrule
      AL
      &(0.44, 0.72)&(0.68, 2)&(0.4, 0.68)&(1, 0.92)&(0, 0.2)&(0.36, 0.52)&(0, 0)& (0.16, 0.32)\\
      \bottomrule
    \end{tabular}
    \end{table}

 SS fails to estimate accurately when the difficulty of task is high (see the right half of Table.1). Under this prior, the results using $\mbox{S}_2\mbox{M}$ (or 2-M) are slightly better than those using other methods, especially HPPM, which leads to large swamping errors. In addition, $\mbox{S}_2\mbox{M} +$  HS out-performs AL.

\subsection{Simulation example 2}
In this example, we compare $\mbox{S}_2\mbox{M}$ and 2-M, when the true coefficient vector $\beta_T$ contains different levels of signal strengths. We set $n = 50$ with $p=300$ and $r=10$. There are three 15s and seven 4s among the  10 signals. The $uncor$ and $cor$ are defined  as before.

The estimation in this simulation example is accurate, regardless of using HS, SS or AL. However, the difference between two levels of signal strength in $\beta_T$ is large enough to cause $\mbox{S}_2\mbox{M}$ and 2-M to have different performances.  $\mbox{S}_2\mbox{M}$ leads to excellent variable selection, while 2-M always masks all the 7 signals of lower signal strength. Again, in this example, $\mbox{S}_2\mbox{M}$ out-performs the competing methods, CS and HT, when the covariates matrix contains correlated covariates.

\begin{table}[htb]

    \caption{Masking and Swamping errors}
    \centering
    \begin{tabular}{ccccccccc}
     \toprule
      $n,p,r$& \multicolumn{2}{c}{$n=50, p=300, r=10$}\\
      & uncor & cor\\
     \midrule
      HS+$\mbox{S}_2\mbox{M}$ &(0,0)&(0.22,0.22)\\
      HS+2M &(7,0)&(7,0)\\
      HS+CS &(0,0)&(1.88,0)\\
      HS+HT &(0,0)&(0.26,0.22)\\
      \midrule
      SS+$\mbox{S}_2\mbox{M}$ &(0,0)&(0.64,0.64)\\
      SS+KM &(7,0)&(7,0)\\
      SS+HPPM &(0,0.5)&(0.6,1.24)\\
      SS+MPM &(0,0.26)&(0.64,0.94)\\
      \midrule
      AL & (0,0)& (0.74,0.76)\\
      \bottomrule
    \end{tabular}
    \end{table}
In addition, with five noise covariates correlated to (all correlations above 0.98) each of the two signal covariates, the two errors (under setting 2) with using HS are (0.46,0.46), (0.46,0.46), (2,0), and (1.16,0.12). $\mbox{S}_2\mbox{M}$ out-performs both CS and HT when there are more correlated covariates.

\subsection{Simulation example 3}
We consider the lymphoma dataset \citep{rosenwald2002use}, which consists of 240 observations and 7399 features representing 4128 genes from the Lymphochip cDNA microarray. We randomly selected 2000 features used as predictors, and randomly selected 30 of them to be signals. The response was simulated by the linear model (\ref{eq:model}), with standardized predictors, a coefficient vector $\beta$, and $\epsilon\sim N(0,1)$.
The coefficients of signal predictors in $\beta$ are 4 and the coefficients of noise predictors are 0.
. Among $2000^2$ pairs of predictors, there are 510 pairs with correlation above 0.8, 380 pairs above 0.90 and 323 pairs above 0.95.
We used the horseshoe prior with the four variable selection methods in the earlier two examples. MCMC was run for 10,000 iterations after discarding a burn-in of 5,000. We obtained the following pairs of errors (1,1), (1,1), (4,0), (2,1), corresponding to $\mbox{S}_2\mbox{M}$, 2-M, CS and HT respectively.

\section{Discussion}\label{sec:disc}
In this article, we developed a simple but useful method for doing variable selection using shrinkage priors by post-processing posterior samples of the regression coefficients.  Our method is essentially applicable to any prior and is based on only one tuning parameter.   We observe excellent performances of our method in terms of computational efficiency and dealing with correlated covariates.
The only tuning parameter associated with this method plays a key role to minimize the chance of masking while keeping the chance of swamping low as well.
Although our current proposal for the tuning parameter works well in most situations we tried,  we would like to explore a theoretically rigorous way to choose this tuning parameter in future.

The theoretical results are restricted to priors with only global shrinkage parameters.  Although we have provided several numerical analysis to better understand the shrinkage properties for the horseshoe prior when the covariates are highly correlated, we aim to study reverse-shrinkage more rigorously for the general class of global-local priors  \eqref{eq:gl}  in future.



\section{appendix}
\begin{lemma}\label{lemma1}
Define function $h$ as
\begin{equation} \label{functionH}
h(x)=\int_{\beta}N(x;\beta,\sigma^2(X'X)^{-1})\pi(\beta)d\beta.
\end{equation}
With normal priors on $\beta$ and $X'X$ as in Theorem 1, the function $h$ can be written as:
\begin{equation}\label{functionh}
h(x_1,x_2)=C\sigma^{-2}\sqrt{\frac{\kappa^2}{1-(1-\kappa)^2\rho^2}}\exp\bigg\{
\frac{1}{2\sigma^2}(f_1x_1^2+f_2x_2^2+2f_3x_1x_2) \bigg\},
\end{equation}\label{h(x,y)}
where $\kappa=1/ (1+\tau^2)$, $C$ is a constant independent from $\sigma, \tau, \rho, x_1, x_2$, and
\begin{equation}\label{f1f2f3}
f_1(\kappa;\rho)=f_2(\kappa;\rho)=\frac{(\rho^2-1-\rho^2\kappa)\kappa}{1-(1-\kappa)^2\rho^2} ,\quad
f_3(\kappa;\rho)=-\frac{\rho\kappa^2}{1-(1-\kappa)^2\rho^2}.
\end{equation}
\end{lemma}
\begin{proof}
Through the calculations below, $C$ represents different constant numbers from step to step. However, $C$ is always independent from $\sigma, \tau, \rho, x_1$ and $x_2$.
\begin{eqnarray*}
\begin{split}
h(x_1,x_2)=\int_{\beta_1,\beta_2} &\frac{1}{2\pi\sigma^2}\exp\bigg\{
-\frac{1}{2\sigma^2}[
(x_1-\beta_1)^2+(x_2-\beta_2)^2+2\rho(x_1-\beta_1)(x_2-\beta_2)]\bigg\}\\
&\frac{1}{(2\pi\sigma^2\tau^2)}\exp\{
-\frac{1}{2\sigma^2\tau^2}(\beta_1^2+\beta_2^2)\} d\beta_1d\beta_2.\\
\end{split}
\end{eqnarray*}
First we try to integrate $\beta_1$ out, obtaining
\begin{equation}\label{star}
h(x_1,x_2)=C\int_{\beta_2}\frac{1}{\sigma^4\tau^2}\exp\begin{Bmatrix}
-\frac{1}{2\sigma^2}(x_1^2+x_2^2+\beta_2^2+\frac{1}{\tau^2}\beta_2^2-2x_2\beta_2+2\rho x_1x_2-2\rho x_1\beta_2)
\end{Bmatrix}I_1d\beta_2,
\end{equation}
where $I_1=\sqrt{2\pi\mu}\sigma\exp\begin{Bmatrix}
\frac{1}{2\sigma^2}B_1
\end{Bmatrix}$,
with $B_1^2=A_1^{-1}(\rho x_2+x_1-\rho \beta_2)^2$, $A_1^2=\frac{\tau^2+1}{\tau^2}$ and $\mu=\frac{\tau^2}{1+\tau^2}$.

Substitute $I_1$ and $B_1^2$ back to (\ref{star}), and continue to integrate (\ref{star}) with respect to $\beta_2$, obtaining
\begin{equation}\label{star2}
h(x_1,x_2)=C\sigma^{-3}\tau^{-2}\exp\begin{Bmatrix}
-\frac{1}{2\sigma^2}(x_1^2+x_2^2+2\rho x_1x_2)
\end{Bmatrix}\exp\begin{Bmatrix}
\frac{u}{2\sigma^2}(\rho x_2+x_1)^2
\end{Bmatrix}I_2,
\end{equation}
where
$I_2=\sqrt{2\pi}\sigma\exp\begin{Bmatrix}
\frac{1}{2\sigma^2}B_2^2
\end{Bmatrix}\sqrt{\frac{\mu}{1-\mu^2\rho^2}}$, with $B_2^2=A_2^{-2}[\mu\rho(\rho x_2+x_1)-(\rho x_1+x_2)]^2$ and $A_2^2=\frac{1}{\mu}-\mu\rho^2=\frac{1-\mu^2\rho^2}{\mu}$.

Substituting $I_2$ with $B_2^2$ back to (\ref{star2}), and letting $\kappa=1-\mu=\frac{1}{1+\tau^2}$, (\ref{functionh}) can be obtained.
\end{proof}


\begin{lemma}\label{lemma2}
Suppose $|\hat\beta_{MLE,1}|=A|\hat\beta_{MLE,2}|$, $\hat\beta_{MLE }=(\hat\beta_{MLE,1},\hat\beta_{MLE,2})$, then the normal estimators can be written as functions of the MLEs as:
\begin{equation}\label{lemma2res}
\begin{split}
&\hat\beta_{N,1}=\begin{bmatrix}
1-\frac{1}{1-\rho^2}\begin{pmatrix}
R_1(\hat\beta_{MLE})-\frac{\rho R_2(\hat\beta_{MLE})}{A}
\end{pmatrix}
\end{bmatrix}\hat\beta_{MLE,1} \\
&\hat\beta_{N,2}=\begin{bmatrix}
1-\frac{1}{1-\rho^2}\begin{pmatrix}
R_1(\hat\beta_{MLE})-\rho R_2(\hat\beta_{MLE}) A
\end{pmatrix}
\end{bmatrix}\hat\beta_{MLE,2}, \\
\end{split}
\end{equation}
where
$R_1(\hat\beta_{MLE})=-\frac{1}{A}(Af_1(\kappa)+f_3(\kappa))$ and
$R_2(\hat\beta_{MLE})=-(f_2(\kappa)+Af_3(\kappa))$.
\end{lemma}
\begin{proof}
Continuing the Lemma \ref{lemma1}, the derivatives of the function $h(x_1,x_2)$ with respective to $x_i$ is
\begin{align*}
\frac{\partial}{\partial x_i}h(x_1,x_2)=C\sigma^{-2}(f_ix_i+f_3x_{3-i})h(x_1,x_2)\ \ \ i=1,2.
\end{align*}
where $C$ here denotes the constant before the exponential in the function $h$. Define
\begin{eqnarray}\label{R1R2}
R_i^*(x_1,x_2)=-\frac{1}{x_i}\frac{\frac{\partial}{\partial x_i}h}{h}=-\frac{1}{\sigma^2x_i}(f_ix_i+f_3x_{3-i})\ \ \ i=1,2.
\end{eqnarray}
Considering $|\hat\beta_{MLE,1}|=A|\hat\beta_{MLE,2}|$, both $R_1^*(\hat\beta_{MLE})$ and $R_2^*(\hat\beta_{MLE})$ can be written as functions of $A$ as in Lemma.
Using the result of the Proposition 1 in \cite{griffin2010inference}, where
\begin{eqnarray*}
S(\hat\beta)=\sigma^2(X'X)^{-1}[R_1^*(\hat\beta_{MLE})\ 0;0\ R_2^*(\hat\beta_{MLE})]=(X'X)^{-1}[R_1(\hat\beta_{MLE}),R_2(\hat\beta_{MLE})]
\end{eqnarray*}
in this case, (\ref{lemma2res}) can be obtained.
\end{proof}

\begin{lemma}\label{lemma3}
Define
\begin{equation}
S_1=\frac{\begin{pmatrix}
R_1-\frac{\rho R_2}{A}
\end{pmatrix}}{1-\rho^2},
S_2=\frac{\begin{pmatrix}
R_2-\rho R_1A
\end{pmatrix}}{1-\rho^2}
\end{equation}
then $-1<f_1(\kappa;\rho)=f_2(\kappa;\rho)<f_3(\kappa;\rho)<0$,
$0<S_1<1,S_2<1$, for any $0<\rho<1$, $A>1$, and $\tau>0$.
\end{lemma}
\begin{proof}
$f_1=f_2$ and $f_3<0$ can be directly obtained from the definitions (\ref{f1f2f3}).
\begin{eqnarray}\label{e1}
f_1(\kappa;\rho)>-1
\Leftrightarrow \frac{(\rho^2-1-\rho^2\kappa)\kappa}{1-(1-\kappa)^2\rho^2}>-1
\Leftrightarrow \rho^2\kappa-\kappa-\rho^2\kappa^2>(\kappa^2+1-2\kappa)\rho^2-1.
\end{eqnarray}
For the last inequality in (\ref{e1}), $LHS-RHS=2\rho^2\kappa(1-\kappa)+(1-\rho^2)(1-\kappa)>0$. Thus $f_1>-1$.
\begin{eqnarray}\label{e2}
f_3>f_1
\Leftrightarrow &(\rho^2-1-\rho^2\kappa)\kappa<-\rho\kappa^2.
\end{eqnarray}
For the inequality in (\ref{e2}), $LHS-RHS=(1-\rho)(\rho\kappa^2-\kappa-\rho)<0$,
this leads to $f_3>f_1$.

Since $f_1,f_2$ and $f_3$ are always negative, $R_1$ and $R_2$ are always positive.
\begin{eqnarray}\label{temp}
S_1>0\Leftrightarrow R_1>\frac{\rho R_2}{A}
\Leftrightarrow \frac{R_1}{R_2}>\frac{\rho}{A}.
\end{eqnarray}
We proved
$\frac{R_1}{R_2}>\frac{A+\rho}{A+\rho A^2}$.
in the following proof of Theorem \ref{Theorem1}. Considering the fact $\frac{A+\rho}{A+\rho A^2}>\frac{\rho}{A}$, the last inequality of (\ref{temp}) is true. Thus $S_1>0$.

We prefer to use $\mu=1-\kappa=\tau^2/[1+\tau^2]$ to prove $S_1<1$:
\begin{eqnarray}\label{e3}
S_1<1\Leftrightarrow R_1-\frac{\rho R_2}{A}<1-\rho^2\Leftrightarrow AR_1-\rho R_2<A-A\rho^2.
\end{eqnarray}
For the last inequality in (\ref{e3}), we can finally obtain $LHS-RHS=(1-\rho^2)(\rho\mu^2-\rho\mu+A\rho^2\mu^2-A\mu)$.
Since both $\rho\mu(\mu-1)$ and $A\mu(\rho\mu-1)$ are negative, $LHS-RHS$ is less than 0.
\begin{eqnarray}\label{e4}
S_2<1\Leftrightarrow R_2-\rho R_1A<1-\rho^2.
\end{eqnarray}
For the inequality in (\ref{e4}),
$LHS-RHS=(1-\rho^2)[(A\rho\mu^2-A\rho\mu)+(\rho^2\mu^2-\mu)]$.
Since both $(A\rho\mu^2-A\rho\mu)$ and $(\rho^2\mu^2-\mu)$ are negative, thus $LHS$ is indeed less than $RHS$.
\end{proof}
Now we are ready to give the proof of Theorem \ref{Theorem1}.
\begin{proof}
Case 1: the Theorem is true when $S_2<0$. That is because $S_1$ is always between 0 and 1, causing $|\hat\beta_{N,1}|$ will always be less than $|\hat\beta_{MLE,1}|$. But $|\hat\beta_{N,2}|$ will be greater than $|\hat\beta_{MLE,2}|$ since $(1-S_2)$ must be greater than 1 in this case.

Case 2: when both $S_1$ and $S_2$ are between 0 and 1, we have
\begin{eqnarray}\label{e5}
(\ref{theorem1res})\Leftrightarrow \frac{S_1}{S_2}>1\Leftrightarrow \frac{R_1-\frac{\rho R_2}{A}}{R_2-\rho R_1A}>1
\Leftrightarrow AR_1-\rho R_2> AR_2-\rho R_1A^2.
\end{eqnarray}

For the last inequality in (\ref{e5}), substituting $R_1$ and $R_2$ derived from Lemma \ref{lemma2}, we can obtain $LHS-RHS=\rho(1-\rho^2)(A^2-1)(1-\kappa)>0$.

Note this verifies $\frac{R_1}{R_2}>\frac{A+\rho}{A+\rho A^2}$, which is needed for proving $S_1>0$ earlier.
\end{proof}
\begin{lemma}\label{hforhs}
The function $h$ is defined as (\ref{functionH}), under the model (\ref{eq:model}) with the horseshoe prior on $\beta$, then
\begin{equation}\label{res3}
h(x_1,x_2)=C\int_{\kappa_1,\kappa_2}F(\kappa_1,\kappa_2;\rho)E(\kappa_1,\kappa_2;\rho,x_1,x_2)d\kappa_1d\kappa_2,
\end{equation}
where $\kappa_i=1/[1+\tau^2\lambda_i^2]$, $i=1,2$, $C$ is a constant independent from $(x_1,x_2,\rho,\lambda_1,\lambda_2)$, and

$F(\kappa_1,\kappa_2;\rho)=[1-(1-\kappa_1)(1-\kappa_2)\rho^2]^{-\frac{1}{2}}[1-(1-\tau^2)\kappa_1]^{-1}[1-(1-\tau^2)\kappa_2]^{-1}(1-\kappa_1)^{-\frac{1}{2}}(1-\kappa_2)^{-\frac{1}{2}}$

$E(\kappa_1,\kappa_2;\rho,x_1,x_2)=\exp\begin{Bmatrix}
\frac{1}{2\sigma^2}(f_1x_1^2+f_2x_2^2+2f_3x_1x_2)
\end{Bmatrix},$

with

$f_i(\kappa_1,\kappa_2;\rho)=[(\rho^2-1-\rho^2\kappa_{3-i})\kappa_i][1-(1-\kappa_1)(1-\kappa_2)\rho^2]^{-1}\ \ \ i=1,2$

$f_3(\kappa_1,\kappa_2;\rho)=-\rho\kappa_1\kappa_2[1-(1-\kappa_1)(1-\kappa_2)\rho^2]^{-1}$.

It follows that the horseshoe estimator can be represented as the right hand side of (\ref{lemma2res}), where $R_i$ is
\begin{equation}\label{res4}
R_i(x_1,x_2)=-\frac{1}{x_i}\frac{\int_{\kappa_1,\kappa_2}[f_i(\kappa_1,\kappa_2)x_i+f_3(\kappa_1,\kappa_2)
x_{3-i}]F(\kappa_1,\kappa_2)E(\kappa_1,\kappa_2)d\kappa_1d\kappa_2}{\int_{\kappa_1\kappa_2}F(\kappa_1,\kappa_2)
E(\kappa_1,\kappa_2)d\kappa_1d\kappa_2}\ \ \ i=1,2.
\end{equation}
\end{lemma}
\begin{proof}
With the horseshoe prior being applied, the $\pi(\beta)$ in (\ref{functionh}) becomes
\begin{align*}
\pi(\beta)=\int_{\lambda}\pi(\beta|\lambda)\pi(\lambda)d\lambda
=\int_{\lambda_1,\lambda_2}N(\beta_1;0,\sigma^2\tau^2\lambda_1^2)N(\beta_2;0,\sigma^2\tau^2\lambda_2^2)\pi(\lambda_1)\pi(\lambda_2)d\lambda_1d\lambda_2
\end{align*}
The function $h$ now is an integral with respect to $(\beta_1,\beta_2,\lambda_1,\lambda_2)$. The trick here is to integrate $\beta_1,\beta_2$ out, then $h$ will be an integral with respect to $\lambda_1,\lambda_2$ only:
\begin{equation}\label{functionh1}
h(x)=\int_{\lambda}\pi(\lambda)\begin{Bmatrix}
\int_{\beta}N(x;\beta,\sigma^2(X'X)^{-1})\pi(\beta|\lambda)d\beta
\end{Bmatrix}d\lambda.
\end{equation}

Note $\pi(\beta|\lambda)$ is a two-dimension normal distribution, with $\kappa_1$ and $\kappa_2$ defined in this Lemma, we can easily integrate $\beta_1,\beta_2$ out through the similar approach as we did in the proof of Lemma \ref{lemma1}, and the integral with respect to $\beta$ inside the braces above is
\begin{eqnarray*}
\int_{\beta}=C'\frac{\sqrt{\kappa_1\kappa_2}}{\sqrt{1-(1-\kappa_1)(1-\kappa_2)\rho^2}}E(\kappa_1,\kappa_2;\rho,x_1,x_2),
\end{eqnarray*}
where $C'$ is a constant independent from $(x_1,x_2,\rho,\lambda_1,\lambda_2)$.

The prior on $\lambda_i$ is
$\pi(\lambda_i)=2/[\pi(1+\lambda_i^2)]$
then the prior on $\kappa_i$ is followed as:
\begin{eqnarray*}
\pi(\kappa_i)=\frac{\tau}{\pi}\frac{1}{1-(1-\tau^2)\kappa_i}(1-\kappa_i)^{-\frac{1}{2}}\kappa_i^{-\frac{1}{2}}.
\end{eqnarray*}
Substituting $\int_{\beta}$ and $\pi(\kappa_i)$ back to (\ref{functionh1}), (\ref{res3}) is obtained. Considering
\begin{eqnarray*}
\frac{\partial}{\partial x_i}h(x_1,x_2)=C\sigma^{-2}\int_{\kappa_1,\kappa_2}(f_ix_i+f_3x_{3-i})F(\kappa_1,\kappa_2)E(\kappa_1,\kappa_2)d\kappa_1d\kappa_2
\end{eqnarray*}
and defining $R_i=\sigma^2R_i^*-(\frac{\partial}{\partial x_i}h)/(x_ih)$, use the result of Proposition 1 in \cite{griffin2010inference} again, (\ref{res4}) is obtained.
\end{proof}

\bibliographystyle{biometrika}
\bibliography{NSF_refs}
\end{document}